\documentclass{article}

 \usepackage[final]{eiml_style_2025}

\usepackage[utf8]{inputenc} 
\usepackage[T1]{fontenc}    
\usepackage{hyperref}       
\usepackage{url}            
\usepackage{booktabs}       
\usepackage{amsfonts}       
\usepackage{nicefrac}       
\usepackage{microtype}      
\usepackage{xcolor}         
\usepackage{amsmath}
\usepackage{amsthm}
\usepackage{graphicx}
\usepackage{caption}    
\usepackage{subcaption} 
\usepackage{natbib} 
\usepackage[linesnumbered,ruled,vlined]{algorithm2e}

\newtheorem{thm}{Theorem}
\newtheorem{Assumption}{Assumption}

\newcommand{\x}{\mathbf{x}}

\newcommand{\X}{\mathbf{X}}

\renewcommand{\P}{{\mathbb{P}}}
\newcommand{\ourmethod}{\texttt{CP4SBI}}

\newcommand{\locartourmethod}{\texttt{CP4SBI}}
\usepackage{amsmath}

\title{A Three-Way Testing Framework for Quantifying Epistemic Calibration Uncertainty in SBI}

\author{%
  Luben M. C. Cabezas \\
  DEs, ICMC\\
  Federal Univ. of São Carlos, Univ. of São Paulo, Brazil \\
  Univ. Grenoble Alpes, Inria, \\
  CNRS, Grenoble INP, LJK, France \\
  \texttt{lucruz45.cab@gmail.com} \\
  \And
  Pedro L. C. Rodrigues \\
Univ. Grenoble Alpes, Inria, \\
CNRS, Grenoble INP, LJK, France \\
  \texttt{pedro.rodrigues@inria.fr} \\
  \And
  Rafael Izbicki \\
  DEs\\
  Federal Univ. of São Carlos, Brazil \\
  \texttt{rafaelizbicki@gmail.com} \\
}

\begin{document}

\maketitle

\begin{abstract}
  Current experimental scientists increasingly rely on simulation-based inference (SBI) to invert complex models with intractable likelihoods. A primary goal in these settings is to obtain credible regions with valid coverage. While recent model-agnostic conformal calibration methods have succeeded in constructing credible sets with prescribed local Bayesian coverage, their approximate nature introduces inherent epistemic uncertainty in the calibration process. In this work, we propose a novel tool for diagnosing calibration uncertainty. Our approach is based on a simple three-way hypothesis testing procedure. We demonstrate how this tool can be used to effectively assess necessary simulation budgets for calibration sets and analyze the epistemic uncertainty associated with cutoff estimation.
\end{abstract}

\section{Introduction}
Modern scientific research relies heavily on complex stochastic simulators which generate synthetic data ($\x \in \mathcal{X}$) from input parameters ($\theta \in \Theta$). These models are ubiquitous across fields like biology \citep{min2017deep, ramirez2024ai}, astrophysics \citep{von2025kids, jeffrey2025dark,carzon2025trustworthy}, and neuroscience \citep{rodrigues2018riemannian, bernardo2024simulation}, but they pose a fundamental inverse problem: inferring the parameter values that best explain observed data \citep{tarantola2005inverse, casella2024statistical}. Since these simulators typically have an intractable likelihood function $p(\x \mid \theta)$, classical Bayesian techniques like MCMC are infeasible. Simulation-based inference (SBI) 
\citep{cranmer2020frontier,masserano2023simulator}
can address this by estimating the posterior distribution $p(\theta \mid \x)$ directly from simulated data, often leveraging deep generative models—such as normalizing flows, score-based diffusion models, and flow matching—for flexible and expressive approximations.

However, despite their empirical success, modern SBI methods face a critical reliability issue: lack of proper calibration. As shown by \citet{hermans2021towards}, these techniques often produce overconfident posterior distributions, leading to credible regions $C(\x_{\text{obs}})$ that fail to achieve their target nominal coverage when constructed from approximate posteriors $\hat{p}(\theta \mid \x_{\text{obs}})$. This miscalibration severely compromises the validity of scientific analyses and downstream decision-making that rely on these uncertainty estimates \citep{murphy2022probabilistic}.

To address this challenge and provide robust uncertainty quantification, certain works have leveraged the conformal prediction (CP) framework \citep{shafer2008tutorial, vovk2022algorithmic} to the SBI context. CP methods offer powerful, distribution-free guarantees on coverage probabilities, ensuring credible regions contain the true parameter with the specified frequency. Recent works exploring this direction include those by \citet{patel2023variational}, \citet{baragatti2024approximate}, and \citet{cabezas2025cp4sbi}. Specifically, the work in \citet{cabezas2025cp4sbi} introduced a novel estimator-agnostic approach capable of calibrating credible regions with local coverage guarantees through tree-based partitions, compatible with both density-based and sample-based SBI estimators.

However, these methods are still approximations, introducing inherent epistemic uncertainty in the calibration step. Specifically, these approaches attempt to approximate the target region {$C(\x) = \{\theta \in \Theta: s(\theta ; \x) < t(\x)\}$}, where $s(\theta ; \x)$ is a non-conformity score (e.g., $s(\theta; \mathbf{x}) = -\hat{p}(\theta \mid \mathbf{x})$, corresponding to the HPD) and the cutoff $t(\x)$ is defined by the conditional coverage requirement $\P(\theta \in C(\X) \mid \X = \x)$.  Since the cutoff $t(\x)$ must be estimated using a finite calibration set, the resulting point estimator, $\hat{t}$, possesses inherent uncertainty. Quantifying this uncertainty is crucial for understanding and diagnosing potential deficiencies in the calibration method itself.

Building upon the method presented by \citet{cabezas2025cp4sbi}, we propose a novel heuristic to quantify and diagnose this calibration epistemic uncertainty using a simple three-way hypothesis testing procedure. Through benchmark examples, we demonstrate that this tool is effective for assessing the simulation budget for the calibration set and for systematically analyzing the epistemic uncertainty associated with the cutoff estimation. Our method thus provides a valuable diagnostic capability for credible region construction in SBI. A flowchart illustrating our approach is presented in Figure \ref{fig:epistemic_CP4SBI_diagram}.

\begin{figure}[htb]
    \centering
    \includegraphics[width=1\linewidth,]{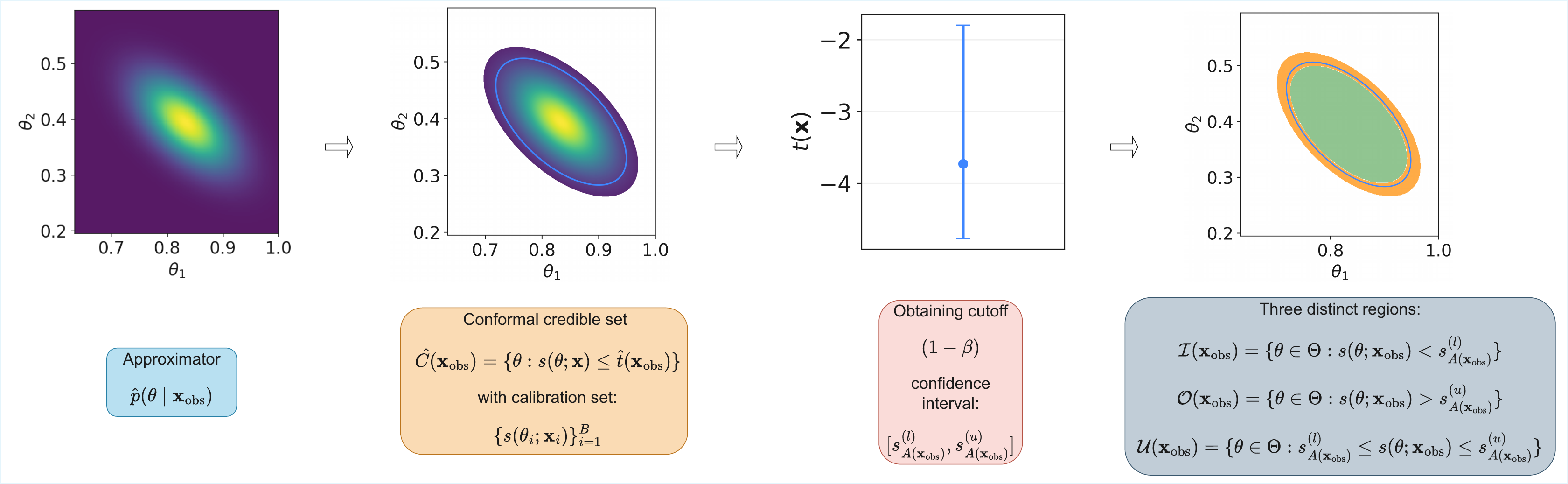}
    \caption{Flowchart of our approach to evaluate calibration epistemic uncertainty on estimated credible sets. We begin with a trained posterior approximator (blue panel) obtained from a fixed training set. Then, using a calibration set, we derive a conformalized credible set $\hat{C}(\mathbf{x}_{\text{obs}}) = \{\theta : s(\theta; \mathbf{x}) \leq \hat{t}(\mathbf{x}_{\text{obs}})\}$ for a specific observation $\mathbf{x}_{\text{obs}}$ (orange panel, blue contour). For this estimated cutoff $\hat{t}(\mathbf{x}_{\text{obs}})$, we derive a non-parametric $(1-\beta)$ confidence interval based on the score's order statistics (red panel). This confidence interval is then used to define three distinct regions (grey panel): inside (light green), undetermined (dark orange), and outside (white), where the size of the undetermined region directly conveys the amount of epistemic uncertainty in the credible set estimation.}
    \label{fig:epistemic_CP4SBI_diagram}
\end{figure}

\section{Relation to Other Work}

\textbf{Posterior distribution approximation.}
Over recent years, various SBI strategies have been developed for posterior estimation when the likelihood is intractable. Initial methods include Approximate Bayesian Computation (ABC) \citep{marin2012approximate, sisson2018handbook}, which approximates the posterior via rejection sampling. These were followed by more efficient amortized methods \citep{cranmer2015approximating, gutmann2016bayesian, izbicki2019abc}: neural posterior estimation (NPE) \citep{papamakarios2016fast, lueckmann2017flexible, greenberg2019automatic, deistler2022truncated}, neural likelihood estimation (NLE) \citep{papamakarios2019sequential, frazier2023bayesian}, and density ratio estimators 
\citep{izbicki2014high, hermans2020likelihood, durkan2020contrastive, dalmasso2020confidence, DalmassoLF2I}. More recently, advanced techniques have adapted implicit generative models, such as flow-matching and diffusion models, to SBI \citep{wildberger2023flow, geffner2023compositional, linhart2024diffusion}. In this work, we specifically leverage NPE as the base estimator for deriving credible regions $C(\x)$. A more thorough review of SBI practices and methods is provided by \citet{deistler2025simulation}.

\textbf{Conformal Prediction applied to SBI.}
Existing work applying standard Conformal Prediction (CP) to SBI primarily provides marginal coverage guarantees \citep{baragatti2024approximate, patel2023variational}.
However, these techniques generally fail to provide locally or approximately conditionally valid credible regions. To address this limitation, \citet{cabezas2024distribution} uses local conformal methods to improve frequentist coverage of these sets, while \locartourmethod{}, introduced by \citet{cabezas2025cp4sbi}, aims to improve Bayesian posterior coverage. 

This locally adaptive partitioning technique originated in the context of supervised learning and prediction \citep{cabezas2025regression} and is seamlessly extended to the SBI context by \citet{cabezas2025cp4sbi}. \locartourmethod{} works by adaptively partitioning the feature space $\mathcal{X}$ using a regression tree trained to predict the nonconformity score $s(\theta;\mathbf{x})$. The derived partition, $\mathcal{A}$, isolates regions with approximately identically distributed scores, enabling a localized and data-adaptive estimation of the cutoff threshold $t(\mathbf{x})$ that yields better approximations of conditionally valid credible sets. Further technical details on the method are provided in Appendix \ref{appendix:locart_method}.

\textbf{Epistemic Uncertainty.}  When analyzing stochastic problems, we distinguish between two sources of uncertainty: aleatoric and epistemic  \citep{hullermeier2021aleatoric,Izbicki2025,cabezasepistemic}. Aleatoric uncertainty is the inherent, irreducible randomness in the data-generating process—specifically, that the observation $\x$ does not uniquely determine the parameter $\theta$, a property encoded by the posterior $p(\theta \mid \x)$.  Consequently, estimating the credible region $C(\x)$  is primarily an attempt to quantify this aleatoric uncertainty. Conversely, \textit{epistemic uncertainty} arises from limitations in data and the resulting lack of knowledge about the true data-generating function or the estimated model. In this context, understanding how the estimator $\hat{C}(\x)$ behaves as a function of the finite calibration data is an example of quantifying and assessing epistemic uncertainty. To the best of our knowledge, this assessment has not been performed in the SBI literature.

\textbf{Three-way hypothesis testing.} In the statistics literature, traditional hypothesis testing operates under a two-decision paradigm concerning a hypothesis $H_0$: either rejecting or failing to reject it \citep{casella2024statistical, degroot1986probability}. In contrast, three-way hypothesis testing offers a crucial third decision: to reject, accept, or remain in doubt about $H_0$ \citep{berg2004no,Izbicki2025REACT}. This third option of ``remaining in doubt'' directly aligns with situations where insufficient data prevents a confident decision, reflecting a state of high epistemic uncertainty. Furthermore, these methods simultaneously control both Type I and Type II errors \citep{esteves2016logical,coscrato2020agnostic}, yielding strong conclusions with a controlled probability of error. Therefore, in our work, we leverage this three-way testing framework for each parameter $\theta \in \Theta$, framing the resulting agnostic region as the region of high epistemic uncertainty whose size directly reports the amount of uncertainty associated with the estimated credible region.









\section{Uncertainty on credible sets}
\label{sec:methods}
In what follows, we assume to have access to a posterior estimator $\hat{p}(\theta \mid \x)$ previously trained by some SBI approach, e.g. neural posterior estimation. Also, we let $\{(\theta_1, \X_1), \ldots, (\theta_B, \X_B)\}$ be a calibration dataset drawn independently from the joint distribution $\pi(\theta)p(\x\mid\theta)$ and not used for estimating $\widehat{p}(\theta\mid\x)$. Concretely, each  $(\theta,\X)$ is drawn by first sampling $\theta$ from a prior distribution $\pi(\theta)$ and then $\X \mid\theta$ from the statistical model $p(\x\mid\theta)$, which is the forward simulator in SBI. Following \cite{cabezas2025cp4sbi}, we fix the conformity score to $s(\theta; \mathbf{x}) = -\hat{p}(\theta \mid \mathbf{x})$.

Our goal is to attempt to recover the target credible region
 \(  C(\x) := \{\theta \in \Theta \mid s(\theta;\X) \leq  t(\x) \}\), where $t(\x)$ is the $1-\alpha$ quantile of the distribution of $s(\theta;\X)|\x$.
 Our starting point is \locartourmethod{} \citep{cabezas2025cp4sbi}, which approximates
$t(\x)$ in the following way.
First, we build a partition of $\mathcal{X}$, $\mathcal A$.
For each observation $\x$, let $A(\x)$ denote the region of the partition $\mathcal{A}$ that contains it. We then define $I_{A(\x)}$ as the set of indices of all calibration points lying in the same region, $I_{A(\x)} = \{b \in \{1,\ldots,B\} : \x_b \in A(\x)\}$. 
Finally, we compute the estimated cutoff $\hat{t}(\x)$  as the adjusted $\alpha$-quantile of the set $T_{A(\x)} = \{s(\theta_b;\X_b) : b \in I_{A(\x)}\}$.

In this section, we describe how we can evaluate the uncertainty about $t(\x)$ that comes from this estimation process, and how to propagate it to the estimated credible region. 
 Let $Z \sim \text{Binomial}(|I_{A(x)}|, 1 - \alpha)$, and choose $l$ to be the largest value such that
 $\mathbb{P}(l - 1 \leq Z) \leq \beta/2$, and $u$ be the smallest value such that
 $\mathbb{P}(Z \geq u) \leq \beta/2$.
 Next, 
let $s_{A(\x)}^{(i)}$ be the $i$-th order statistic of the scores  in $T_{A(\x)}$.
 An approximate   $(1 - \beta)$-level confidence set for the unknown cutoff $t(\x)$ is given by
\begin{align}
[s_{A(\x)}^{(l)}, s_{A(\x)}^{(u)}].
\end{align}


Finally, we propagate this uncertainty about $t(\x)$ to the estimated confidence sets by defining three distinct regions based on the parameter score $s(\theta; \x)$:
\begin{align*}
  &\mathcal{I}(\x)=\{\theta \in \Theta: s(\theta;\x) \leq  s_{A(\x)}^{(l)} \},  \\
  &\mathcal{O}(\x)=\{\theta \in \Theta: s(\theta;\x) \geq  s_{A(\x)}^{(u)} \}, \text{ and }\\
  &\mathcal{U}(\x)=\{\theta \in \Theta : s_{A(\x)}^{(l)} \leq s(\theta;\x) \leq  s_{A(\x)}^{(u)} \}.
\end{align*}
This structure directly implements a three-way hypothesis testing framework \citep{esteves2016logical}. Here, $\mathcal{I}(\x)$ is the \textbf{acceptance region} (parameters confidently \textbf{inside}), $\mathcal{O}(\x)$ is the \textbf{rejection region} (parameters confidently \textbf{outside}), and $\mathcal{U}(\x)$ is the \textbf{agnostic region} (parameters whose status remains \textbf{uncertain} due to cutoff estimation error). The existence of $\mathcal{U}(\x)$ explicitly quantifies the epistemic uncertainty in the boundary estimation; it contains parameters whose status is uncertain solely due to the finite size of the calibration set used to estimate the cutoff. The size of $\mathcal{U}(\x)$ is therefore directly controllable by increasing the calibration simulation budget $B$.

Theorem \ref{thm:coverage} formalizes how these regions quantify the probability of drawing incorrect conclusions. It makes the following assumption.
\begin{Assumption}
\label{ass:coverage}
For fixed $\x \in \mathcal{X}$, the statistics $s_{A(\x)}^{(l)}$ and $s_{A(\x)}^{(u)}$ satisfy
\[
\P\big(t(\x) \le s_{A(\x)}^{(l)} \big) \le \beta/2
\quad \text{and} \quad
\P\big(t(\x) \ge s_{A(\x)}^{(u)} \big) \le \beta/2,
\]
where the probability is taken with respect to the randomness of the calibration set.
\end{Assumption}
In Appendix \ref{appendix:proofs} we explain why such assumptions are approximately valid for our problem.

\begin{thm}
\label{thm:coverage}
Fix $\x \in \mathcal{X}$ and $\theta \in \Theta$.  
Let $(s_{A(\x)}^{(l)}, s_{A(\x)}^{(u)})$ satisfy  Assumption \ref{ass:coverage}. 
Then:
\begin{itemize}
    \item If $\theta \notin C(\x)$, we have $\P\big(\theta \in \mathcal{I}(\x)\big) \le \beta/2$;
    \item If $\theta \in C(\x)$, we have $\P\big(\theta \in \mathcal{O}(\x)\big) \le \beta/2$,
\end{itemize}
where the probability is taken with respect to the randomness of the calibration set.
\end{thm}

Theorem \ref{thm:coverage} guarantees that the probability of making an incorrect decision—whether wrongly accepting an external parameter or wrongly rejecting an internal one—is controlled by $\beta/2$ in each direction, despite using finite calibration data. Crucially, the agnostic region $\mathcal{U}(\x)$ isolates and identifies those parameters for which this residual epistemic uncertainty remains unresolved.


\section{Experiments}
\label{sec:experiments}
We evaluate our method on two categories of simulation models: the now standard SBI benchmark from \citet{lueckmann2021benchmarking} and a neural mass model (NMM), which is a complex non-linear model from computational neuroscience consisting of a system of stochastic differential equations describing the generation of neural activity on a cortical column~\citep{jansen1995electroencephalogram}. Specifically, we use the implementation detailed by \citet{rodrigues2021hnpe}.

For the SBI benchmarks, we compare our method's performance across three different calibration budget sizes ($B \in \{1000, 2000, 4000\}$) to assess how epistemic uncertainty behaves as the budget changes. We also analyze the method's convergence by tracking how the estimated inside region and its boundary progressively align with the target region—the HPD derived from the true posterior samples available for each benchmark. We highlight one key result in the main text and detail the remainder in Appendix \ref{appendix:exp_details}. For the NMM application, we use a single, large calibration budget ($B = 8000$) to assess the epistemic uncertainty in that complex example.

All further details, including the posterior approximator architecture, training procedures, additional results, and \locartourmethod{} hyperparameters, are provided in Appendix \ref{appendix:exp_details}.

\subsection{SBI benchmarks results}
For the SBI benchmarks, we limit our analysis to two-dimensional parameter spaces with a specific $\x_{\text{obs}}$. Therefore, whenever a benchmark task has parameter dimension greater than two, we restrict our analysis by fitting the NPE approximator exclusively to the first two parameters. Figure \ref{fig:uncertainty_comparison_gaussian_linear_uniform} showcases this specific application using the \textit{Gaussian Linear Uniform} example.
\begin{figure}[!http]
    \centering
    \includegraphics[width=1.0\textwidth]{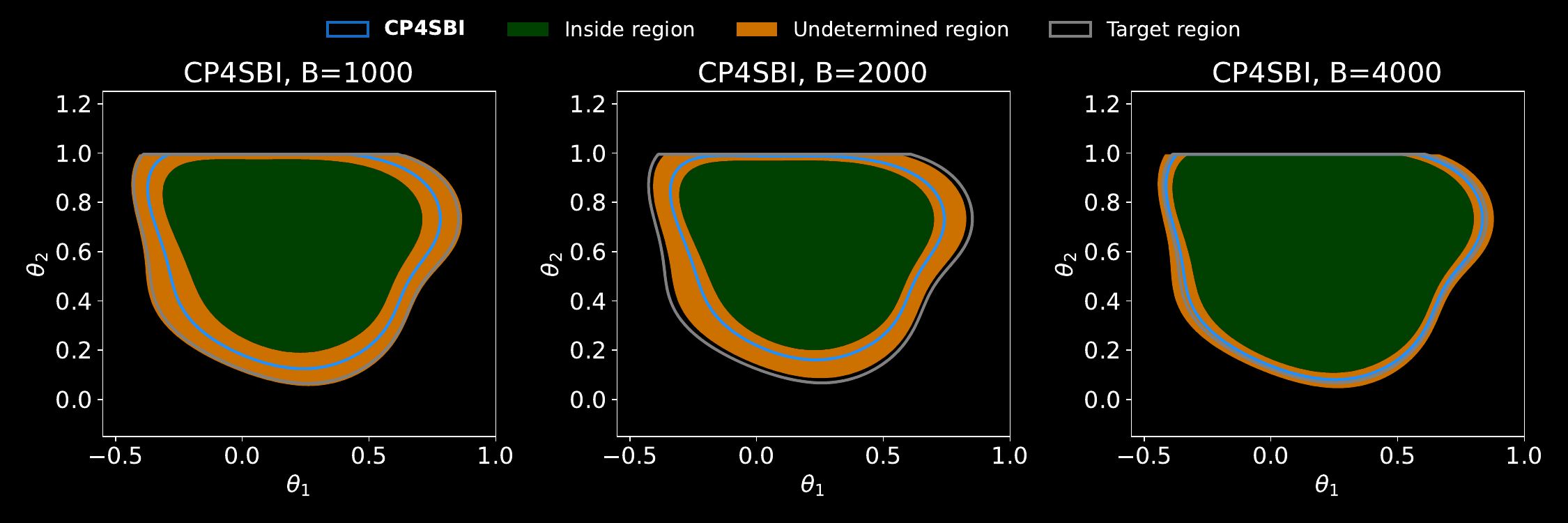}
    \caption{\footnotesize Epistemic uncertainty quantification on the Gaussian Linear Uniform example for a specific $\x_{\text{obs}}$. The areas shaded in dark green and dark orange represent, respectively, the inside (low epistemic uncertainty) and the undetermined (high epistemic uncertainty) region, while the contours in light blue and grey represent, respectively, the \locartourmethod{} and target region. As the calibration budget ($B$) increases, the size of the dark orange area decreases, confirming the expected reduction in epistemic uncertainty.}
    \label{fig:uncertainty_comparison_gaussian_linear_uniform}
\end{figure}

As the calibration budget increases, the undetermined region shrinks, reflecting the expected reduction in epistemic uncertainty when larger calibration sets are used to estimate the credible region. At the same time, the \locartourmethod{} boundary (light blue) and the confident interior region progressively align with the target boundary (gray), indicating improved estimation accuracy.

Similar trends are observed across all other simulators (Figures \ref{fig:easy_medium_epistemic_uncertainty}–\ref{fig:hard_epistemic_uncertainty}, Appendix \ref{appendix:additional_exp}), where the undetermined region consistently contracts as the calibration budget $B$ grows.

\subsection{Application on NMM}
Our experiments assess the three parameters of the NMM by analyzing each unique pair separately, which involves fitting a dedicated NPE approximator and estimating a credible region for every pair. As done for SBI benchmarks, we also consider a specific $\x_{\text{obs}}$ generated by fixed values of the parameters of interest. Figure \ref{fig:uncertainty_regions_jrnmm_largest_budget_row} showcases the application of our method for this simulator, considering a large calibration budget $B = 8000$.
\begin{figure}[http]
    \centering
    \includegraphics[width=1\linewidth]{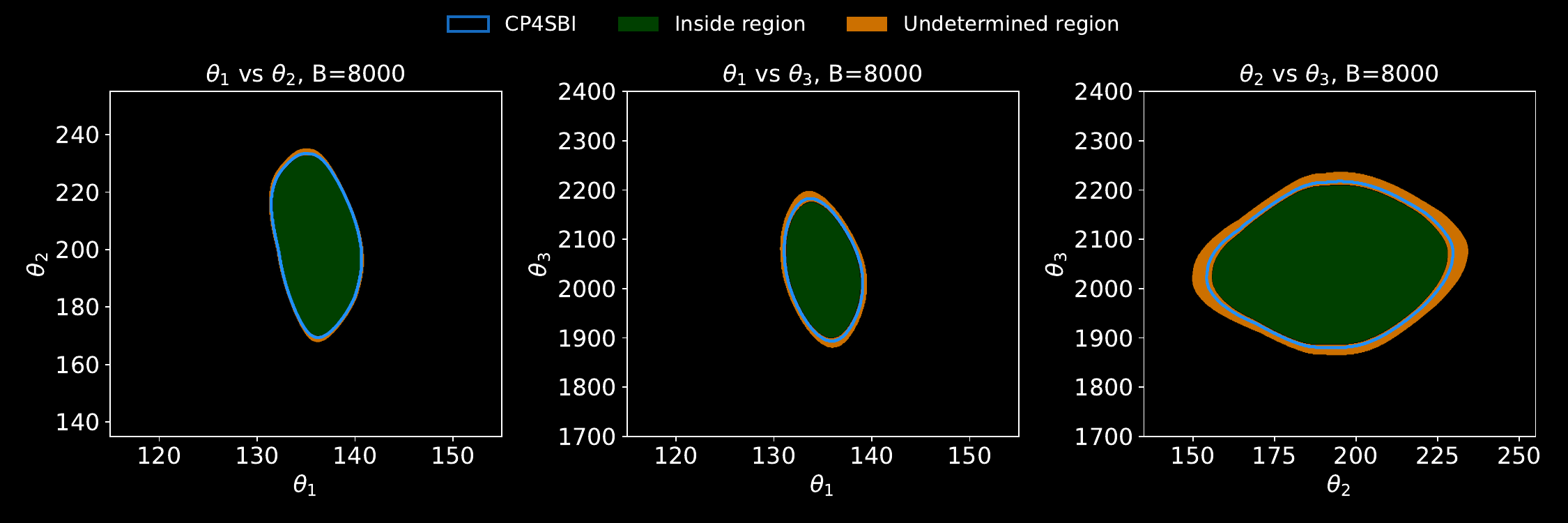}
    \caption{\footnotesize Epistemic uncertainty across parameter pairs in the NMM. The figure shows the estimated regions for three different parameter pairs given a fixed observation $\x_{\text{obs}}$. The dark green areas are confidently inside the credible set, while the dark orange areas highlight the agnostic region (epistemic uncertainty). A clear distinction is seen in uncertainty levels: the first pair has a minimal uncertain region, the second shows moderate uncertainty, and the final pair displays the largest agnostic region.}
    \label{fig:uncertainty_regions_jrnmm_largest_budget_row}
\end{figure}

We observe that the level of epistemic uncertainty varies significantly across the parameter pairs. The first pair, $(\theta_1, \theta_2)$ displays small uncertainty, the second, $(\theta_1, \theta_3)$,  shows medium uncertainty, and the final pair, $(\theta_2, \theta_3)$, exhibits the largest uncertain region. This difference highlights the varying difficulty of estimation across dimensions. Consequently, the latter two combinations warrant more attention, suggesting a possible necessity to increase the calibration budget specifically to derive more precise credible regions for $(\theta_1, \theta_3)$ and $(\theta_2, \theta_3)$.

\section{Conclusions}
In this work, we presented a novel method for representing epistemic uncertainty in the context of credible set estimation for simulation-based inference (SBI). Building upon the work of \citet{cabezas2025cp4sbi} and the three-way hypothesis testing, our approach constructs an undetermined region ($\mathcal{U}$). This region explicitly delineates parameters whose status (inside or outside the target credible set) is highly uncertain due to estimation error, alongside well-defined inside ($\mathcal{I}$) and outside ($\mathcal{O}$) regions where the method is highly confident about the parameter status.

We also show that both $\mathcal{I}$ and $\mathcal{O}$ have controlled probabilities for incorrect conclusions, which provides  theoretical grounding for our epistemic uncertainty analysis. Furthermore, we demonstrate the utility of this diagnostic tool across SBI benchmarks, showing that the undetermined region size effectively tracks the necessary calibration budget, $B$. We apply this to a complex stochastic neural mass model to highlight how epistemic uncertainty varies significantly across different parameter dimensions for real-world data.

For future work, we intend to expand this diagnostic tool and type of analysis to other credible set estimators beyond \locartourmethod{}, providing similar theoretical and empirical applications. Specifically, we aim to apply this method to a wider range of real-world problems, utilizing it as a robust tool for selecting optimal simulation budgets for calibration. The code implementation for our method and experiments is publicly available at \href{https://github.com/Monoxido45/three_way_CP4SBI}{https://github.com/Monoxido45/three\_way\_CP4SBI}.

\section*{Acknowledgments}

This study was financed in part by the Coordenação de Aperfeiçoamento de Pessoal de Nível Superior - Brasil (CAPES)
- Finance Code 001. L.M.C.C is grateful for the fellowship
provided by São Paulo Research Foundation (FAPESP),
grants 2022/08579-7 and 2025/06168-8. PLCR was supported by a national grant managed by the French National Research Agency (Agence Nationale de la Recherche) attributed to the SBI4C project of the MIAI AI Cluster, under the reference ANR-23-IACL-0006. RI is grateful for the financial support of FAPESP (grant 2023/07068-1) and CNPq (grants 305065/2023-8 and 403458/2025-0).

\bibliography{references}


\appendix

\section{Technical details on \locartourmethod{}}
\label{appendix:locart_add_details}
This section details the \locartourmethod{} algorithm's procedure and outlines the necessary computational and implementation details used throughout our experiments.

\subsection{Method details}
\label{appendix:locart_method}
The primary motivation behind \locartourmethod{} is to derive an optimal partition of the feature space $\mathcal{X}$. This partition groups observations ($\mathbf{x}$) that share the same conditional distribution of the conformal score, $s(\theta ; \mathbf{x}) \mid \mathbf{X} = \mathbf{x}$. By achieving this, \locartourmethod{} ensures that within each resulting partition element $A \in \mathcal{A}$, the average coverage probability, $\mathbb{P}(\theta \in C(\mathbf{x}) \mid \mathbf{x} \in A)$, closely approximates the desired local conditional coverage, $\mathbb{P}(\theta \in C(\mathbf{x}) \mid \mathbf{X} = \mathbf{x})$.

To derive this partition, \locartourmethod{} leverages the fact that regression-tree partitions can group similar score distributions in practice \citep{meinshausen2006quantile}. Specifically, it fits a regression tree to the calibration data, using the input features $\mathbf{x}_i$ to predict their corresponding nonconformity scores $s_i$. The final terminal leaves of this fitted tree then define the partition $\mathcal{A}$ that approximates the ideal conditional groupings. This process is fully detailed in Algorithm \ref{alg:cp4sbi_locart}, which is adapted from \citet{cabezas2025cp4sbi}.

\begin{algorithm}[ht]
\caption{\locartourmethod{} algorithm}
\label{alg:cp4sbi_locart}
\KwIn{
    Calibration set $\mathcal{D} = \{(\theta_i, \x_i)\}_{i=1}^B$,
    posterior approximator $\widehat{p}(\theta \mid \x)$,
    conformity score $s(\theta; \X)$,
    nominal level $\alpha$,
    new observation $\x_{\text{obs}}$
}

\textbf{Step I: Score computation} \\
\Indp
1: For each $(\theta_i, \x_i) \in \mathcal{D}$, compute $s_i = s(\theta_i; \x_i)$ using $\widehat{p}$. \\
2: Form the scored dataset $\mathcal{D}' = \{(s_i, \x_i)\}_{i = 1}^B$, then randomly split it into a partition set $\mathcal{D}'_{\text{part}}$ and a cutoff set $\mathcal{D}'_{\text{cut}}$ subsets. \\
\Indm

\textbf{Step II: Partition learning} \\
\Indp
1: Fit a regression tree on $\mathcal{D}'_{\text{part}}$ to predict $s$ from $\x$. \\
2: Use the tree to induce a partition $\mathcal{A} = \{A_1, \dots, A_K\}$ of the feature space. \\
3: Define a region mapping $T: \mathcal{X} \rightarrow \mathcal{A}$ such that $T(\x) = A_j$ if $\x \in A_j$. \\
\Indm

\textbf{Step III: Local quantile estimation} \\
\Indp
1: For each region $A_j \in \mathcal{A}$, define the set of calibration indices $I_j = \{i \mid \x_i \in A_j,\; (s_i, \x_i) \in \mathcal{D}'_{\text{cut}}\}$. \\
2: Compute the local cutoff $t_j$ as the empirical $(1 + 1/|I_j|)(1 - \alpha)$-quantile of $\{s_i\}_{i \in I_j}$. \\
\Indm

\textbf{Step IV: Credible region construction} \\
\Indp
1: Assign $\x_{\text{obs}}$ to region $A_k = T(\x_{\text{obs}})$. \\
2: Return the credible region:
\[
R_{\text{CP4SBI}}(\x_{\text{obs}}) = \left\{ \theta \mid s(\theta; \x_{\text{obs}}) \leq t_k \right\}
\]
\Indm

\KwOut{Credible region $R_{\text{CP4SBI}}(\x_{\text{obs}})$ with marginal and local $1 - \alpha$ coverage}
\end{algorithm}

\subsection{Computational details}
\label{appendix:comp_locart_details}
To derive the partitions for \locartourmethod{} using regression trees, we followed the default procedures from \citep{cabezas2025cp4sbi} and employed two primary stopping criteria: the minimum number of samples per leaf and post-pruning via cost-complexity methods (\textit{ccp\_alpha}) to balance partition complexity and predictive performance. For the small calibration budget ($B=1000$), we set the minimum leaf size to $150$. For all larger budgets ($B=\{2000,4000,8000\}$), we increased this minimum to $300$. All other regression tree hyperparameters were set to the default values provided by the \textit{scikit-learn} library \citep{pedregosa2011scikit}.

Additionally, in accordance with the results of \citet{cabezas2025cp4sbi}, we omit the internal splitting of the calibration set when both training the regression tree and computing local cutoffs. Although splitting is theoretically necessary for strict marginal coverage guarantees, the lack of this internal split has been shown to have a minimal empirical impact on coverage while avoiding a reduction in the sample size available for local cutoff estimation. This choice is therefore maintained across all specified budgets ($B$) to preserve practical performance.

\section{Experiments details and additional results}
\label{appendix:exp_details}
This appendix provides all further experimental details. We specify the NPE model architectures and training budgets and present additional results for the SBI benchmarks.
\subsection{NPE architectures and training budgets}
\label{appendix:npe_architectures}
The Neural Posterior Estimation (NPE) was implemented using conditional normalizing flows \citep{greenberg2019automatic, papamakarios2021normalizing} via the \texttt{sbi} package \citep{BoeltsDeistler_sbi_2025}. For the SBI Benchmarks, we tailored the training budgets to task complexity: small tasks used $B_{\text{train}} = 5000$, medium tasks used $B_{\text{train}} = 10000$, and harder tasks used $B_{\text{train}} = 20000$. These specific budgets are listed in Table \ref{tab:budgets_used}. The architecture utilized was the default \texttt{sbi} implementation, based on Masked Autoregressive Flows (MAF) \citep{papamakarios2017masked}. For the NMM Experiment, we utilized a fixed training budget of $B_{\text{train}} = 10000$. An alternative architecture, the Neural Spline Flow (NSF) \citep{durkan2019neural}, was employed, trained for a maximum of $250$ epochs with a batch size of $100$ and early stopping set to $20$.
\begin{table}[http]
\caption{Training budget used in each task}
\label{tab:budgets_used}
\centering
\begin{tabular}{cc}
\hline
\textbf{Benchmarks}     & \textbf{Training Budget} \\ \hline
Bernoulli GLM           & 10000                    \\
Gaussian Linear         & 5000                     \\
Gaussian Linear Uniform & 5000                     \\
Gaussian Mixture        & 10000                    \\
Lotka-volterra          & 20000                    \\
Two Moons               & 5000                     \\
SIR                     & 20000                    \\
SLCP                    & 20000                    \\ \hline
\end{tabular}
\end{table}

\subsection{Additional results for SBI benchmarks}
\label{appendix:additional_exp}
We visualize all remaining epistemic uncertainty quantification results for the SBI benchmarks in Figure \ref{fig:easy_medium_epistemic_uncertainty} (easy and medium tasks) and Figure \ref{fig:hard_epistemic_uncertainty} (hard tasks).

Figure \ref{fig:easy_medium_epistemic_uncertainty} reveals distinct behavior between task difficulties: for easy tasks (Gaussian Linear and Two Moons), the initial epistemic uncertainty is minimal and quickly vanishes, with inside region remaining relatively stable. In contrast, for medium tasks (Gaussian Mixture and Bernoulli GLM), the initial uncertainty starts slightly larger and decreases more gradually, accompanied by more noticeable changes in the inside region's shape. In common, we notice that for both difficulties, the epistemic uncertainty decreases as the calibration budget $B$ increases.

Concerning the harder tasks, Figure \ref{fig:hard_epistemic_uncertainty} reveals three distinct patterns: the SLCP task shows an initial increase in the undetermined region before it ultimately diminishes; the SIR task starts with slight epistemic uncertainty which decreases rapidly; and the Lotka-Volterra task begins with high uncertainty, which decreases gradually as a function of $B$. Despite these behavioral differences, all tasks exhibit a clear improvement: the interior region's format changes and converges toward the target boundaries as the calibration budget increases. In all these complex settings, the common and key takeaway is the clear reduction of epistemic uncertainty as more calibration data is added to the credible set method.

\begin{figure}[htb]
    \centering
    \begin{subfigure}[b]{1\textwidth}
        \centering
        \includegraphics[width=0.775\textwidth]{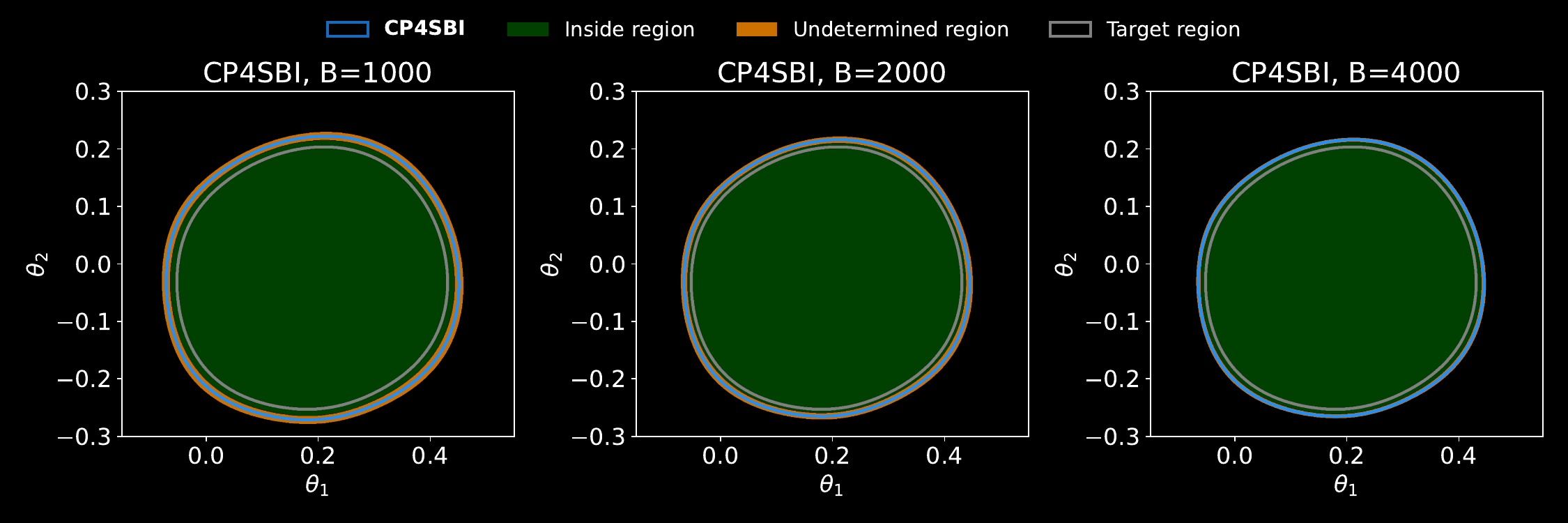} 
        \caption{Gaussian Linear task}
        \label{fig:gaussian_linear_sets}
    \end{subfigure}
    \hfill 
    \begin{subfigure}[b]{1\textwidth}
        \centering
        \includegraphics[width=0.775\textwidth]{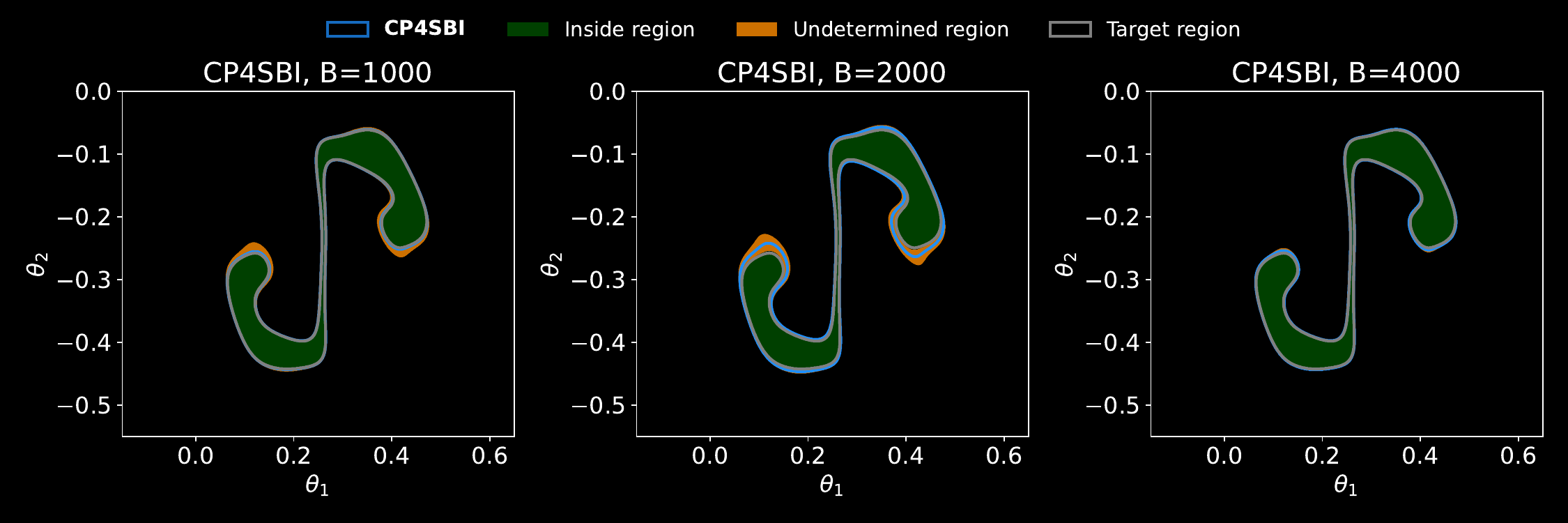} 
        \caption{Two Moons task.}
        \label{fig:two_moons_sets}
    \end{subfigure}
    \hfill
    \begin{subfigure}[b]{1\textwidth}
        \centering
        \includegraphics[width=0.775\textwidth]{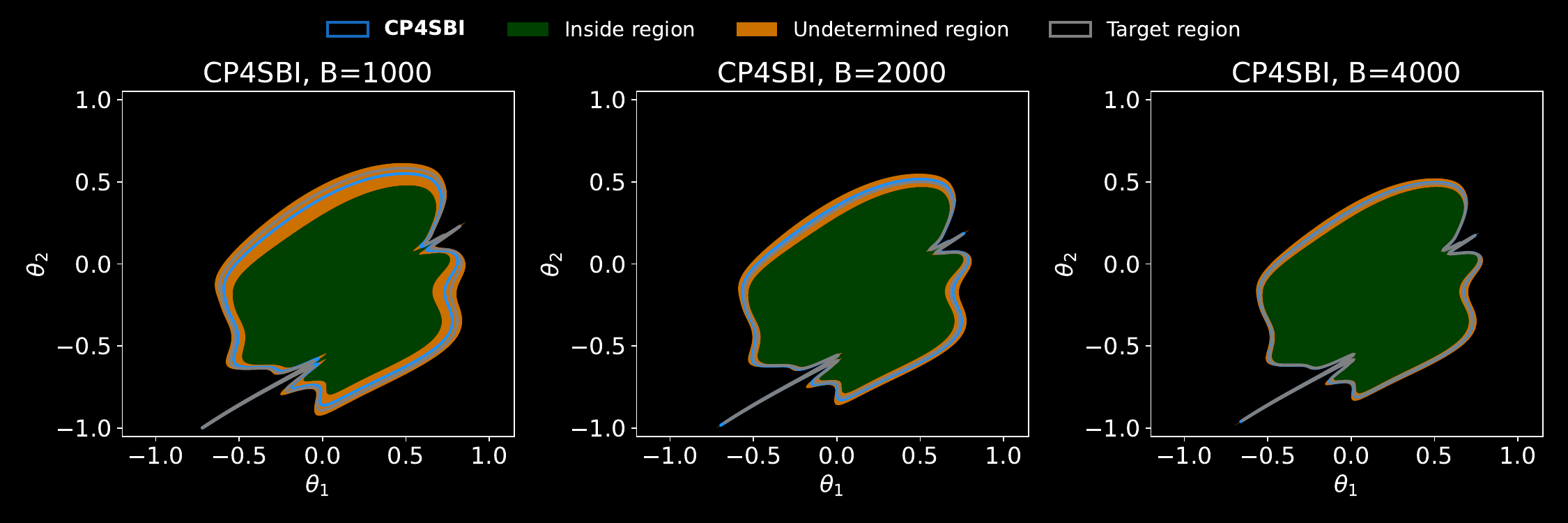} 
        \caption{Gaussian Mixture task.}
        \label{fig:gaussian_mixture_sets}
    \end{subfigure}
    \begin{subfigure}[b]{1\textwidth}
        \centering
        \includegraphics[width=0.775\textwidth]{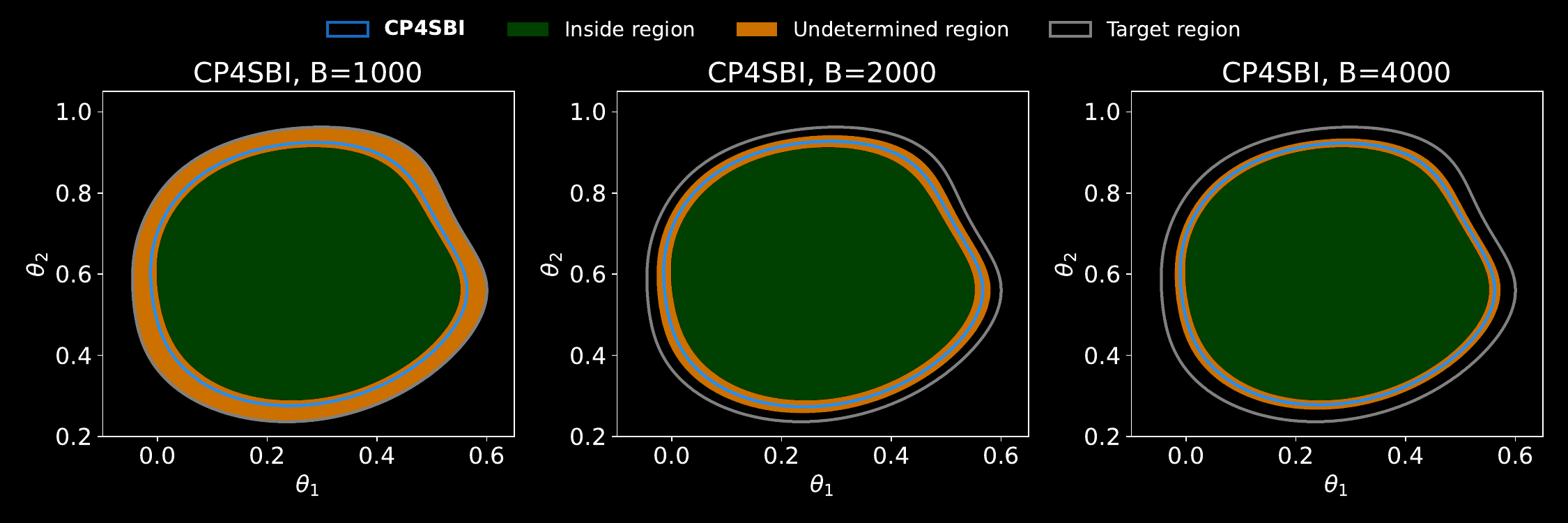} 
        \caption{Bernoulli GLM task.}
        \label{fig:bernoulli_glm}
    \end{subfigure}
    \caption{Epistemic Uncertainty Quantification on Easier/Medium SBI Benchmarks. The dark green and dark orange areas represent the confident interior (low epistemic uncertainty) and the undetermined region (high epistemic uncertainty), respectively. Contours show the \locartourmethod{} estimate (light blue) and the target region (grey). We consistently observe that for every task shown, the epistemic uncertainty (dark orange area) decreases as the calibration budget ($B$) increases.}
    \label{fig:easy_medium_epistemic_uncertainty}
\end{figure}

\begin{figure}[htb]
    \centering
    \begin{subfigure}[b]{1\textwidth}
        \centering
        \includegraphics[width=0.775\textwidth]{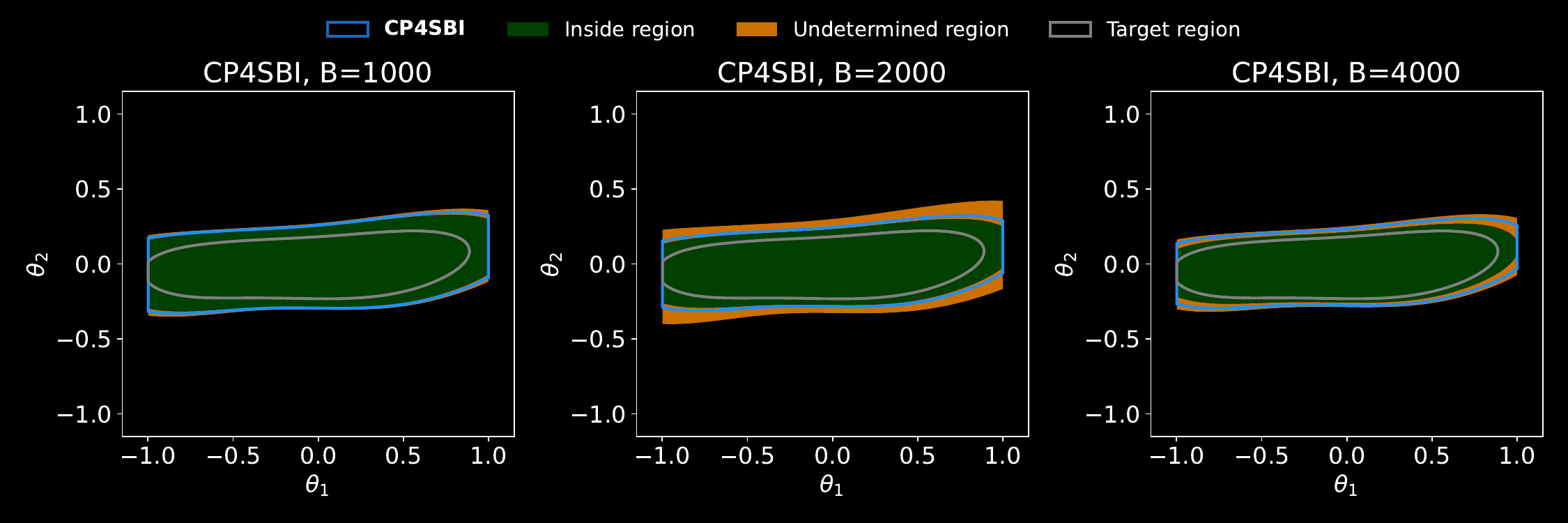} 
        \caption{SLCP task.}
        \label{fig:slcp_sets}
    \end{subfigure}
    \hfill 
    \begin{subfigure}[b]{1\textwidth}
        \centering
        \includegraphics[width=0.775\textwidth]{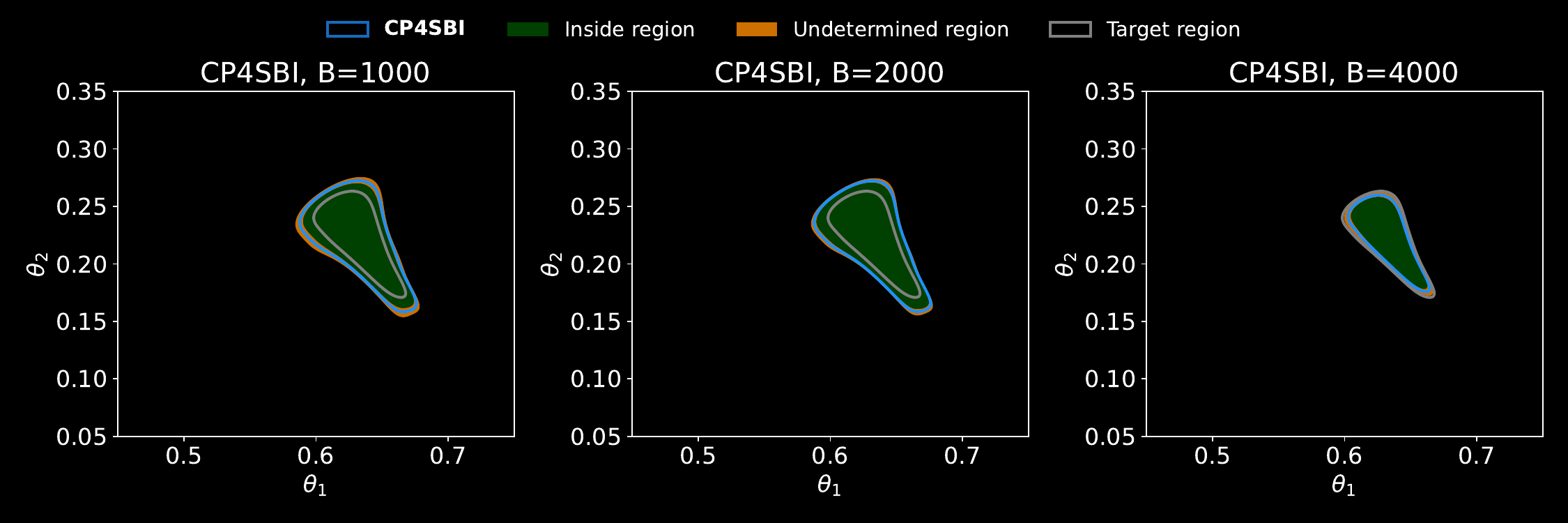} 
        \caption{SIR task.}
        \label{fig:sir_sets}
    \end{subfigure}
    \hfill
    \begin{subfigure}[b]{1\textwidth}
        \centering
        \includegraphics[width=0.775\textwidth]{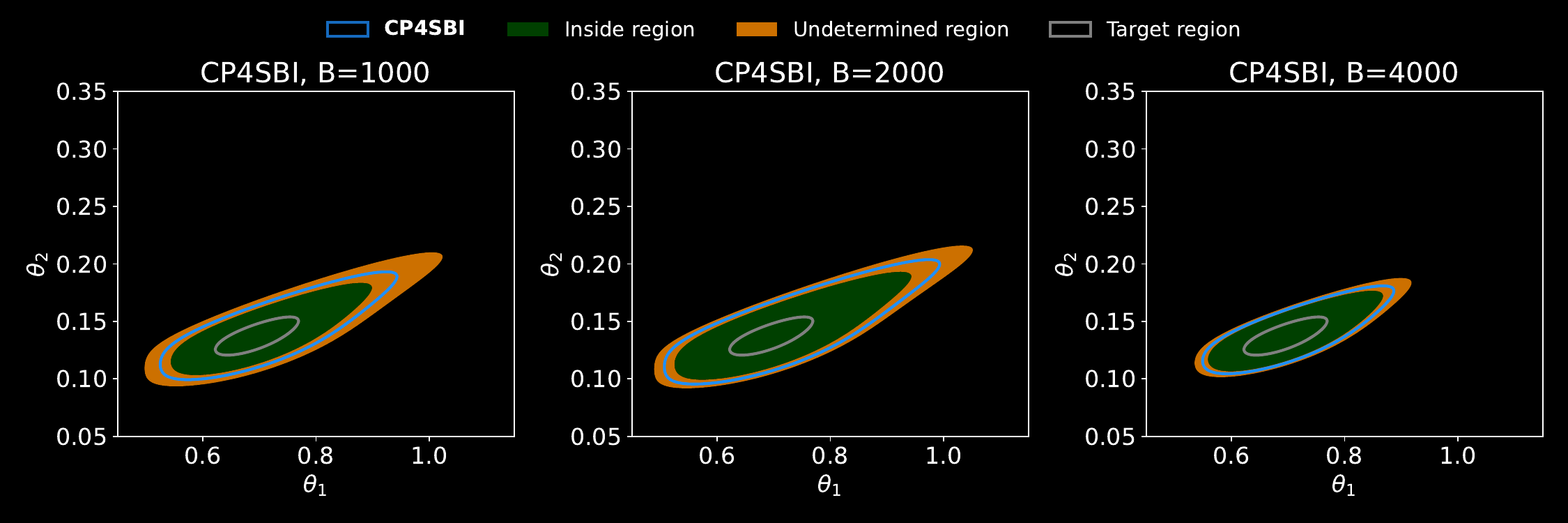} 
        \caption{Lotka-Volterra task.}
        \label{fig:lotka_volterra_sets}
    \end{subfigure}
    \caption{Epistemic Uncertainty Quantification on Harder SBI Benchmarks. The dark green and dark orange areas represent the confident interior (low epistemic uncertainty) and the undetermined region (high epistemic uncertainty), respectively. Contours show the \locartourmethod{} estimate (light blue) and the target region (grey). We confirm the general tendency: the epistemic uncertainty visibly decreases as the calibration budget ($B$) increases.}
    \label{fig:hard_epistemic_uncertainty}
\end{figure}

\section{Proofs}
\label{appendix:proofs}
In this appendix we detail the formal proof of the theoretical result stated in Section \ref{sec:methods} along with the motivation behind Assumption \ref{ass:coverage} and the reason why it is a plausible assumption to make in our context.

\subsection{Approximate validity of Assumption \ref{ass:coverage}}
Let $m=|I_{A(\x)}|$ and let
\[
W=\sum_{b\in I_{A(\x)}} \mathbf{1}_{\{s(\theta_b; \X_b)\le t(\x)\}}.
\]
Conditionally on $\X\in A(\x)$, the indicators are independent. Moreover, by the construction of \ourmethod's partition, they are also approximately identically distributed with
\[
p:=\P\big(s(\theta;\X)\le t(\x)\mid \X\in A(\x)\big)\approx 1-\alpha,
\]
so $W\approx Z$, where $Z\sim\mathrm{Binomial}(m,1-\alpha)$. Since $s^{(1)}_{A(\x)}\le\cdots\le s^{(m)}_{A(\x)}$ are the order statistics of $T_{A(\x)}$,
\[
\{t(\x)\le s^{(l)}_{A(\x)}\}\iff \{W\le l-1\}.
\]
By the definition of $l$ (chosen so that $\P(Z\le l-1)\le \beta/2$),
\[
\P\big(t(\x)\le s^{(l)}_{A(\x)}\big)=\P(W\le l-1)\ \approx\ \P(Z\le l-1)\ \le\ \beta/2.
\]
Analogously, with $u$ chosen so that $\P(Z\ge u)\le \beta/2$,
\[
\P\big(t(\x)\ge s^{(u)}_{A(\x)}\big)=\P(W\ge u)\ \approx\ \P(Z\ge u)\ \le\ \beta/2.
\]

\subsection{Proof of Thereom \ref{thm:coverage}}
\begin{proof}
Fix $\x \in \mathcal{X}$ and $\theta\in\Theta$. Note that $s(\theta;\x)$ is deterministic (given $\x,\theta$); the only randomness comes from the calibration set through $s^{(l)}_{A(\x)}$ and $s^{(u)}_{A(\x)}$.

\medskip
\noindent{(i) If $\theta\notin C(\x)$,} then $s(\theta;\x) > t(\x)$. Moreover, the event $\{\theta\in\mathcal{I}(\mathbf{x})\}$ is $\{s(\theta;\mathbf{x}) \le s^{(l)}_{A(\mathbf{x})}\}$. Since $s(\theta;\mathbf{x}) > t(\mathbf{x})$, it follows that
\[
\{s(\theta;\mathbf{x}) \le s^{(l)}_{A(\mathbf{x})}\} \subseteq \{t(\mathbf{x}) \le s^{(l)}_{A(\mathbf{x})}\}.
\]
Hence, by Assumption~\ref{ass:coverage},
\[
\mathbb{P}\big(\theta\in\mathcal{I}(\mathbf{x})\big) \le \mathbb{P}\big(t(\mathbf{x}) \le s^{(l)}_{A(\mathbf{x})}\big) \le \beta/2.
\]

\medskip
\noindent{(ii) If $\theta\in C(\mathbf{x})$,} then $s(\theta;\mathbf{x}) \le t(\mathbf{x})$. Moreover, the event $\{\theta\in\mathcal{O}(\mathbf{x})\}$ is $\{s(\theta;\mathbf{x}) \ge s^{(u)}_{A(\mathbf{x})}\}$. Since $s(\theta;\mathbf{x}) \le t(\mathbf{x})$, it follows that
\[
\{s(\theta;\mathbf{x}) \ge s^{(u)}_{A(\mathbf{x})}\} \subseteq \{t(\mathbf{x}) \ge s^{(u)}_{A(\mathbf{x})}\}.
\]
Therefore, by Assumption~\ref{ass:coverage},
\[
\mathbb{P}\big(\theta\in\mathcal{O}(\mathbf{x})\big) \le \mathbb{P}\big(t(\mathbf{x}) \ge s^{(u)}_{A(\mathbf{x})}\big) \le \beta/2.
\]
\end{proof}


\end{document}